\documentclass[sigconf]{sig-alternate-per}

\usepackage{booktabs} 
\usepackage{gensymb}
\usepackage{graphicx}
\usepackage{url}
\usepackage{amsmath}
\usepackage{latexsym,amssymb,epsfig,color,verbatim}
\usepackage{algorithm}
\usepackage{subcaption}
\usepackage{algpseudocode}
\usepackage{pifont}
\usepackage{lipsum}
\usepackage{mwe}
\usepackage{slashbox}
\usepackage{color}
\usepackage{cancel}
\usepackage{multirow}
\usepackage{bbm}
\newcommand{\BEQA}{\begin{eqnarray}}
\newcommand{\EEQA}{\end{eqnarray}}

\newtheorem{lemma}{Lemma}
\newtheorem{conjecture}{Conjecture}
\newtheorem{theorem}{Theorem}

\newtheorem{definition}{Definition}
\newtheorem{remark}{Remark}
\newtheorem{proposition}{Proposition}
\usepackage{slashbox}

\begin{document}
\title{On the Analysis of Spatially Constrained \\Power of Two Choice Policies\thanks{This research was sponsored by the U.S. Army Research Laboratory and the U.K. Defence Science and Technology Laboratory under Agreement Number W911NF-16-3-0001 and by the NSF under grant NSF CNS-1617437. The views and conclusions contained in this document are those of the authors and should not be interpreted as representing the official policies, either expressed or implied, of the U.S. Army Research Laboratory, the U.S. Government, the U.K. Defence Science and Technology Laboratory. This document does not contain technology or technical data controlled under either the U.S. International Traffic in Arms Regulations or the U.S. Export Administration Regulations.}}
\author{Nitish K. Panigrahy$^\dagger$, Prithwish Basu$^*$, Don Towsley$^\dagger$, Ananthram Swami$^\mathsection$ and Kin K. Leung$^{**}$
 \\ {\normalsize $^\dagger$University of Massachusetts Amherst, MA, USA. Email: \{nitish, towsley\}@cs.umass.edu}
 \\ {\normalsize $*$Raytheon BBN Technologies, Cambridge, MA 02138, USA. Email: prithwish.basu@raytheon.com}
 \\ {\normalsize $^\mathsection$ Army Research Laboratory, Adelphi, MD 20783, USA. Email: ananthram.swami.civ@mail.mil}
 \\ {\normalsize $^{**}$Imperial College London, London SW72AZ, UK. Email: kin.leung@imperial.ac.uk}
}

\maketitle

\begin{abstract}
We consider a class of power of two choice based assignment policies for allocating users to servers, where both users and servers are located on a two-dimensional Euclidean plane.  In this framework, we investigate the inherent tradeoff between the communication cost, and load balancing performance of different allocation policies. To this end, we first design and evaluate a \emph{Spatial Power of two} (sPOT) policy in which each user is allocated to the least loaded server among its two geographically nearest servers sequentially. When servers are placed on a two-dimensional square grid, sPOT maps to the classical Power of two ({\it POT}) policy on the Delaunay graph associated with the Voronoi tessellation of the set of servers. We show that the associated Delaunay graph is $4$-regular and provide expressions for asymptotic maximum load using results from the literature. For uniform placement of servers, we map sPOT to a classical balls and bins allocation policy with bins corresponding to the Voronoi regions associated with the second order Voronoi diagram of the set of servers. We provide expressions for the lower bound on the asymptotic expected maximum load on the servers and prove that sPOT does not achieve POT load balancing benefits. However, experimental results suggest the efficacy of sPOT with respect to expected communication cost. Finally, we propose two non-uniform  server sampling based POT policies that achieve the best of both the performance metrics. Experimental results validate the effectiveness of our proposed policies.
\end{abstract}

\section{Introduction}\label{sec:intro}

Recent advances in technologies of smart devices, sensors and  embedded processors have enabled the deployment of a large number of computational  and storage resources in a physical space. We collectively call such a set of resources/ services a \emph{distributed service network}. In a distributed service network, for example, Internet of Things \cite{atzori10}, these resources are heavily accessed by various end users/applications that are also distributed across the physical space.

An important design problem in such service networks is the assignment of users/applications to appropriate servers/ resources. Often servers have limited computational capabilities and can only serve a handful number of users. Hence for such a system with very large number of users and servers,  an appropriate server selection strategy typically involves minimizing the maximum number of users assigned to a server, also known as the maximum \emph{load.} While the optimal server selection problem can be solved centrally, due to scalability concerns, it is often preferred to relegate the assignment task to individual users themselves. This interpretation leads to formulating a randomized \emph{load balancing} problem for the distributed service network with the goal to make the overall user-to-server assignment as fair as possible. Many previous works \cite{Adler1998, mitzenmacher96} have used randomization as an effective tool to develop simple and efficient load balancing algorithms in non-geographic settings. A randomized load balancing algorithm can be described as a classical balls and bins problem as follows.

In the classical balls-and-bins model of randomized load balancing, $m$ balls are placed sequentially into $n$ bins. Each ball samples $d$ bins uniformly at random and is allocated to the bin with the least load with ties broken arbitrarily. It is well known that when $d = 1$ and $m = n,$ this assignment policy results in a maximum load of $O(\log n/\log \log n)$ with high probability. However, if $d = 2$, then the maximum load is $O(\log \log n)$ w.h.p. \cite{Azar99}. Thus, there is an exponential improvement in performance from $d = 1$ to $d = 2$. This policy with $d = 2$ is widely known as \emph{Power of Two} (POT) choices and the improvement in maximum load behavior is known as \emph{POT benefits} \cite{mitzenmacher96}. 

While classical balls and bins based randomized load balancing can directly be used for user-to-server assignment in a geographic setting, it is oblivious to the spatial distribution of servers and users. The spatial distribution of servers and users in a service network is vital in determining the overall performance of the service. For example, the Euclidean distance between a user request and its allocated server, also known as \emph{request distance}, directly translates to communication latency incurred by a user when accessing the service. Also, in wireless networks, signal attenuation is strongly coupled to request distance, therefore developing allocation policies to minimize request distance can help reduce energy consumption. Thus the following natural question arises. \emph{How do we design a load balancing policy that captures the spatial distribution of users and servers?}

In this work, we aim to answer this question. To this end we propose a spatially motivated POT policy: \emph{spatial POT} (sPOT) in which each user is sequentially allocated to the least loaded server among its two geographically nearest servers. We assume both users and servers are placed in a two-dimensional Euclidean plane. When both servers and users are placed uniformly at random in the Euclidean plane, we map sPOT to a  classical balls and bins allocation policy with bins corresponding to the Voronoi regions associated with the second order Voronoi diagram of the set of servers. We show that sPOT performs better than POT in terms of average request distance. However, a lower bound analysis on the asymptotic expected maximum load for sPOT suggests that POT load balancing benefits are not achieved by sPOT. Inspired by the analysis of sPOT, we further propose two assignment policies and empirically show that these policies achieve the best of both request distance and maximum load behavior. 

Our contributions are summarized below:
\begin{enumerate}
\item Analysis of sPOT yielding lower bound expressions for asymptotic expected maximum load. 
\begin{itemize}
\item When servers are placed on a two-dimensional grid, we model sPOT using POT on the Delaunay graph associated with the Voronoi tessellation of the set of servers.
\item When users and servers are placed uniformly at random on a 2-D region, we model {\it sPOT} as classical balls and bins allocation policy with bins corresponding to the Voronoi regions associated with the second order Voronoi diagram of the set of servers.
\end{itemize}
\item Introduction of two non-uniform server sampling based POT policies to improve load and request distance behavior.
\begin{itemize}
\item A candidate set based policy that samples $k$ nearest servers from a user and applies POT on the candidate set.
\item A non-uniform distance decaying sampling based POT in which each user samples two servers with probability inversely proportional to square of the distance to the servers.
\end{itemize}
\item A simulation study validating the effectiveness of the proposed policies.
\end{enumerate}

The rest of this paper is organized as follows. The next section contains some technical preliminaries. In Section \ref{sec:spot} we formally analyze the load behavior of sPOT for different server placement settings. In Section \ref{sec:kspot}, we propose two non-uniform  server sampling based POT policies that achieve both better load and request distance behavior. Finally, the conclusion of this work and potential future work are given in Section \ref{sec:con}.
\section{Technical Preliminaries}\label{sec:tp}

We consider a distributed service network where users and servers are located on a two-dimensional Euclidean plane $\mathcal{D}$.
\subsection{Users and Servers}
\noindent {\bf The users:} Users in the service network are denoted as the set $R$ with cardinality $|R| = m.$ We assume users are placed on a two-dimensional euclidean plane uniformly at random.\\

\noindent {\bf The servers:} In a service network, each user is assigned to a server from the server set $S$ with  cardinality $|S| = n.$ We consider two cases for placing the servers on a two-dimensional euclidean plane, (i) Grid Placement: servers are placed on a square grid topology embedded in Euclidean space $\mathbb{R}^2.$ (ii) Uniform placement: servers are placed uniformly at random on the euclidean plane.

We consider the case where $m=n$. We define the following user allocation strategies.
\subsection{User Allocation Policies}
A user-to-server allocation policy is defined as a mapping $\pi$ such that $\pi: R \to S.$ In this paper, our goal is to analyze the performance of various allocation policies defined in the literature as follows.
\begin{itemize}

\item {\bf Power of One ({\it POO}):} This policy assigns each user to one of the servers chosen uniformly at random from $S$.

\item {\bf Power of Two ({\it POT}):} In this policy, sequentially\footnote{We assume that the users can communicate among themselves in order to agree on a sequential order and such communication cost between users is negligible.} each user samples two servers uniformly at random from $S$ and is allocated to the least loaded server.

\item {\bf Spatial Power of One ({\it sPOO}):} This policy assigns each user to its geographically nearest server.
\end{itemize}
We also propose new policies to reduce both the maximum load  and expected request distance. We define them as follows.
\begin{itemize}
\item {\bf Spatial Power of Two ({\it sPOT}):} Each user is sequentially allocated to the least loaded server among its two geographically nearest servers. 

\item {\bf Candidate set based sPOT ({\it k-sPOT}):} Each user uniformly at random samples two servers from a candidate set consisting of its $k$ geographically nearest servers and is assigned to the leastly loaded server.

\item {\bf Decay based Power of Two ({\it dPOT}):} This policy sequentially considers each user $r_i,$ samples two servers from $S$ (without replacement), each with probability proportional to $1/d_{ip}^{2}$. Here, $d_{ip}$ denotes the euclidean distance between user $r_i$ and server $s_p.$ $r_i$ is then assigned to the server with the least load.
\end{itemize}

\subsection{Performance Metric}
To evaluate and characterize the performance of various allocation policies, we consider the maximum asymptotic ($m \to \infty$) load across all servers and expected request distance as performance metrics.
\begin{definition}
Load on a server, $s \in S$, under allocation policy $\pi$, is defined as $L_s = |\{r|\pi(r)=s, \; \text{with}\; r \in R\}|.$
\end{definition} 	 
\begin{definition}
Request distance of a user request $r \in R$ under allocation policy $\pi$, is defined as $D_r = d_{r\pi(r)}.$ Here $d_{rs}$ denotes the euclidean distance between a user $r$ and a server $s.$
\end{definition}

\subsection{Geometric Structures}
We define the following geometric structures that are useful constructs for analyzing various user allocation policies. 
\subsubsection{Voronoi Diagram} 
A Voronoi cell around a server $s\in S$ is the set of points in $\mathcal{D}$ that are closer to $s$ than to any other server in $S\setminus\{s\}$ \cite{bash07}. The Voronoi diagram $V_S$ of $S$ is the set of Voronoi cells of servers in $S$.
\subsubsection{Delaunay Graph}
The Delaunay graph, $G_S(X,E),$ associated with $S$ is defined as follows. Assign the vertex set $X = S$ and add an edge between servers $u$ and $v$, i.e. $e = (u, v) \in E$ only if the Voronoi cells of $u$ and $v$ are adjacent.
\subsubsection{Higher order Voronoi diagram}
A $p^{th}$ order Voronoi diagram, $H_S^{(p)},$ is defined as partition of $\mathcal{D}$ into regions such that points in each region have the same $p$ closest servers in $S$.

\subsection{Majorization}
We present a few definitions and basic results associated with majorization theory that we apply to analyze sPOT policy later in Section \ref{sub:uniform}.

\begin{definition}
The vector $x$ is said to majorize the vector $y$ (denoted $x\succ y$) if
\begin{align}
&\sum\limits_{i=1}^k x_{[i]} \ge \sum\limits_{i=1}^k y_{[i]}, \; k = 1,\cdots,n-1,\nonumber\\
\text{and} &\sum\limits_{i=1}^n x_{[i]} = \sum\limits_{i=1}^n y_{[i]}
\end{align}
where $x_{[i]} (\text{or } y_{[i]})$ is the $i^{th}$ largest element of $x (\text{or } y)$.
\end{definition}

\begin{definition}
A function $f : R^n \rightarrow R$ is called Schur-convex if
\begin{align}
x \succ y \implies f(x) \ge f(y)
\end{align}
\end{definition}

Consider the following proposition (Chapter 11, \cite{marshall79})
\begin{proposition}\label{prop:multnom}
Let $X$ be a random variable having the multinomial distribution
\begin{align}
\Pr[X = x] = {n \choose {x_1,\cdots,x_n}} \prod\limits_{i=1}^{n} p_i^{x_i}
\end{align}
where $x = (x_1, . . . , x_n) \in \{z : z_i \; \text{are nonnegative integers}, \sum z_i = n\}.$ If $\phi$ is a Schur-convex function of $X$, then $\psi(p) = E_{p}\phi(X)$ is a Schur-convex function of $p$.
\end{proposition}

\section {Spatial Power of Two policy}\label{sec:spot}

\begin{figure*}[htbp]
\centering
\begin{minipage}{.25\textwidth}
\centering
\includegraphics[width=\linewidth]{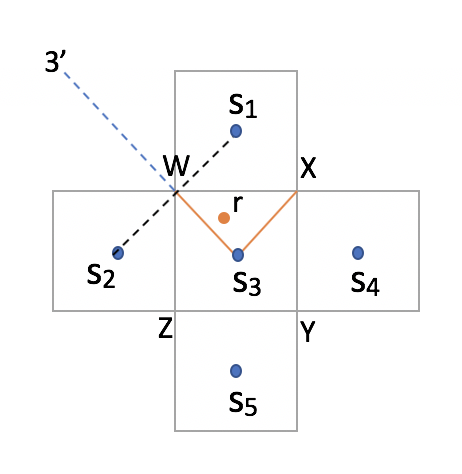}
\caption{Second nearest region for user $r$}
\label{fig:sub_reg}
\end{minipage}\hfill
\begin{minipage}{.3\textwidth}
\centering
\includegraphics[width=\linewidth]{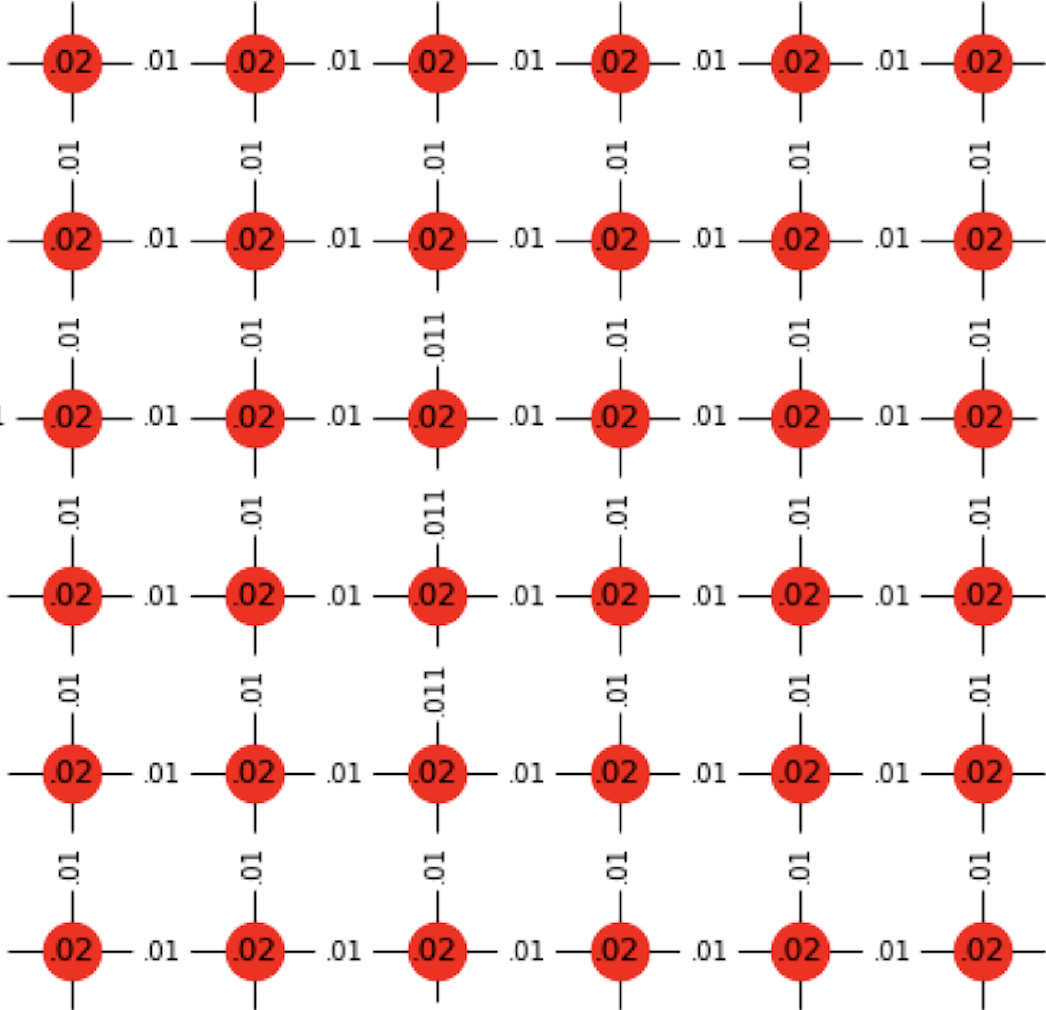}
\caption{Delaunay Graph associated with grid based server placement}
\label{fig:del_grid}
\end{minipage}\hfill
\begin{minipage}{.3\textwidth}
\centering
\includegraphics[width=\linewidth]{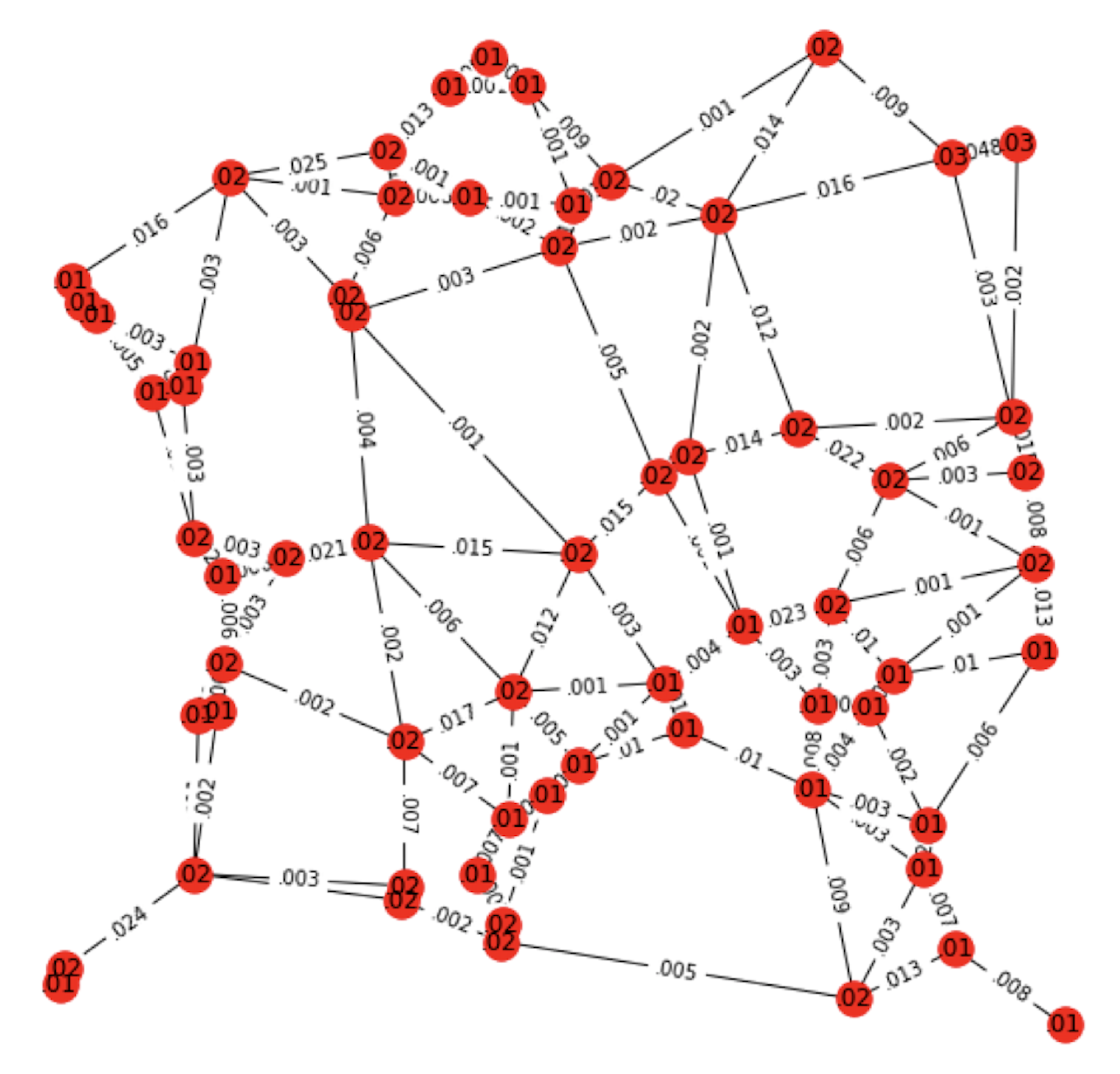}
\caption{Delaunay Graph associated with uniform server placement}
\label{fig:del_uniform}
\end{minipage}
\vspace{-0.5cm}
\end{figure*}

We now analyze the load behavior of sPOT policy for various server placement settings. We assume users are placed uniformly at random on $\mathcal{D}.$

\subsection{sPOT with Grid based server placement}
Consider the case where servers are placed on a two dimensional square grid: $\sqrt{n} \times \sqrt{n}$ on $\mathcal{D}$ with wrap-around.  
Let $B(\{s_1,s_2\},r)$ be the event that the two nearest servers of $r$ are in $\{s_1, s_2\}.$  We prove the following Lemma.

\begin{lemma}\label{lem:equalprobs}
Let $G_S(X,E)$ denote the Delaunay graph associated with $S$ when servers are placed on a two-dimensional square grid. Then
\begin{equation}\label{eq:bfinal}
    \Pr[B(\{s_i,s_j\},r)]=
\begin{cases}
     \frac{1}{|E|},& (s_i, s_j) \in E;\\
     0,& \text{otherwise}.
\end{cases}
\end{equation}
\end{lemma}
\begin{proof}
Let $A(s,r,l)$ denote the event that a random user $r \in R$ is $l^{th}$ closest to $s\in S$ among all servers in $S.$ Denote $NN(r)$ as the geographically nearest server of $r.$ Thus we have 
\begin{align}
\Pr[B(\{s_i,s_j\},r)] = &\Pr[A(s_i,r,1)]\Pr[A(s_j,r,2)|NN(r) = s_i] \nonumber\\
&+ \Pr[A(s_j,r,1)]\Pr[A(s_i,r,2)|NN(r) = s_j].\label{eq:two_near_prob}
\end{align}

It is not difficult to show that all Voronoi cells in $V_S$ have equal areas. As $\Pr[A(s,r,1)]$  is proportional to the area of the Voronoi cell surrounding $s,$ we have 
\begin{align}\label{eq:nn}
\Pr[A(s,r,1)] = 1/|S|\;\forall\;s\in S.
\end{align}

Without loss of generality ({\it W.l.o.g.}) consider a user $r$ placed uniformly at random on $\mathcal{D}$ as shown in Figure \ref{fig:sub_reg}. Denote $\bigtriangleup ABC$ as the triangle associated with vertices $A, B$ and $C.$ Let $NN(r) = s_3.$ We now evaluate $\Pr[A(s_1,r,2)|NN(r) = s_3].$ Clearly, $\Pr[A(s_1,r,2)|\allowbreak NN(r) = s_3]\propto$ Area($\bigtriangleup WXs_3$). We also have 
\begin{align}
\text{Area}(\bigtriangleup WXs_3) &= \text{Area}(\bigtriangleup WZs_3) = \text{Area}(\bigtriangleup YXs_3)\nonumber\\
& = \text{Area}(\bigtriangleup ZYs_3),\nonumber
\end{align}
Therefore $\Pr[A(s_i,r,2)|NN(r) = s_3],$ for $i \in \{1, 2, 4, 5\}$  are all equal. Let $NG(s)$ be the set of neighboring servers of a server $s \in S$ on the square grid. Thus, we have
\begin{equation}\label{eq:cond2}
    \Pr[A(s_j,r,2)|NN(r) = s_i]=
\begin{cases}
     \frac{1}{4},& s_j \in NG(s_i);\\
     0,& \text{otherwise},
\end{cases}
\end{equation}
Note that when $s_j \in NG(s_i),$ the Voronoi cells corresponding to $s_i$ and $s_j$ share an edge. In this case, by definition $(s_i, s_j) \in E.$ Combining Equations \eqref{eq:nn} and \eqref{eq:cond2} and substituting in Equation \eqref{eq:two_near_prob} we get 
\begin{equation}\label{eq:bfinal2}
    \Pr[B(\{s_i,s_j\},r)]=
\begin{cases}
     \frac{1}{2|S|},& (s_i, s_j) \in E;\\
     0,& \text{otherwise},
\end{cases}
\end{equation}
Also, when servers are placed on a square grid, $G_S(X,E)$ is $4$- regular. Thus the total number of edges is $|E| = 2|X| =2|S|.$ Substituting $|S| = |E|/2$ in Equation \eqref{eq:bfinal2} gives us the expression \eqref{eq:bfinal} and completes the proof.
\end{proof}

We consider the following lemma presented in \cite{kenthapadi06}.
\begin{lemma}\label{lm:graph}
Given a $\Delta$-regular graph with $n$ nodes representing $n$ bins, if $n$ balls are thrown into the bins by choosing a random edge and placing into the smaller of the two bins connected by the edge, then the maximum load is at least $\Omega(\log\log n +\frac{\log n}{\log(\Delta\log n)})$ with high probability of $1-1/n^{\Omega(1)}.$
\end{lemma}
\noindent We now prove the following theorem.
\begin{theorem}\label{th:grid}
Suppose servers are placed on a two dimensional square grid : $\sqrt{n} \times \sqrt{n}$ on $\mathcal{D}$ with wrap-around. Let users be placed independently and uniformly at random on $\mathcal{D}.$ Under sPOT, the maximum load over all servers is at least $\Omega(\frac{\log n}{\log\log n})$ with high probability of $1-1/n^{\Omega(1)}.$
\end{theorem}
\begin{proof}
Suppose we map the set of servers to the bins and the users to the balls. The delaunay graph $G_S$ is $4$-regular. Let $e = (s_i, s_j)$ be an edge in $G_S.$ From Lemma \ref{lem:equalprobs}, it is clear that each user (ball) selects an edge $e$ with probability $1/|E|$ (i.e. uniformly at random) and gets allocated to the smaller of the two servers (bins) connected by $e$ under sPOT policy. Thus a direct application of Lemma \ref{lm:graph} with $\Delta = 4$ proves the theorem.
\end{proof}
We verify the results in Lemma \ref{lem:equalprobs} through simulation for a 2D square grid under sPOT policy as shown in Figure \ref{fig:del_grid}. We assign $n=64$ and empirically compute $\Pr[B(\{s_i,s_j\},r)]$ and denote it as edge probability on edge $e$ on the Delaunay graph. We also verify the $\Pr[A(s,r,1)]$ in expression \eqref{eq:nn} and denote it as vertex probability on the Delaunay graph. It is clear from Figure \ref{fig:del_grid} that the edge probabilities are almost all equal and so are the vertex probabilities.
\begin{remark}
Note that, Theorem \ref{th:grid} ensures that no POT benefit is observed when servers are placed on a two dimensional square grid.
\end{remark}
\begin{remark}
Note that Theorem \ref{th:grid} applies to other grid graphs such as a triangular grid, i.e. no POT benefit is observed when servers are placed on a two dimensional triangular grid. The delaunay graph corresponding to a triangular grid based server placement is $6$- regular.
\end{remark}
\subsection{sPOT with Uniform server placement}\label{sub:uniform}
We now consider the case when both users and servers are placed uniformly at random on $\mathcal{D}.$
We no longer can invoke Lemma \ref{lm:graph}. This is due to the fact that the Delaunay graph associated with the servers is no longer regular. Also, the edge sampling probabilities $\Pr[B(s_i,s_j,r)]$ are no longer equal. This is evident from our simulation results on the corresponding Delaunay graph as shown in Figure \ref{fig:del_uniform}. We have $n=64$ servers placed randomly in a 2D square and empirically compute $\Pr[B(\{s_i,s_j\},r)]$ and denote it as edge probability on edge $e$ on the Delaunay graph. Note that the edge probabilities, i.e. $\Pr[B(\{s_i,s_j\},r)]$, are all completely different from each other. Also the Delaunay graph is not regular. Thus we resort to using the second order Voronoi diagram to analyze the  maximum asymptotic load behavior.

Consider the second order Voronoi diagram: $H_S^{(2)}$ associated with the set of servers $S.$ We have the following Lemma [Chapter 3.2, \cite{okabe92}].
\begin{lemma}\label{lm:cellsinhv2}
The number of Voronoi cells in $H_S^{(2)}$ under uniform server placement is upper bounded by $O(3n)$ .
\end{lemma}
We also have the following Lemma.
\begin{lemma}\label{lm:majorization}
Consider the following modified version of balls and bin problem. Suppose there are $n$ balls and $n$ bins. Each ball is thrown into one of the bins according to a probability distribution $p = (p_1, \cdots\allowbreak, p_n)$ with $p_i$ being the probability of each ball falling into bin $i$, in an independent manner. 
Denote $Z$ to be the random variable associated with the maximum number of balls in any bin. The we have
\begin{align}
E_p[Z] \ge k_0\frac{\log n}{\log\log n}\; \text{as}\; n\to\infty.
\end{align}
\noindent where $k_0$ is a scalar constant.
\end{lemma}
\begin{proof}
Denote $X_i$ as the random variable associated with the load for  bin $i.$ Clearly $X = [X_1, X_2, \cdots, X_n]$ follows multinomial distribution
\begin{align}
\Pr[X = x] = {n \choose {x_1,\cdots,x_n}} \prod\limits_{i=1}^{n} p_i^{x_i}
\end{align}
We have $Z = \max(X_1,\cdots,X_n) = \phi(X).$ Clearly, $\phi(x) = \max(x)$ is a schur convex function since $\max(x) = x_{[1]}$ and if $x\succ y$ then $x_{[1]} \ge y_{[1]}$. Also, we have (Chapter 1, \cite{marshall79}): $(p_1,p_2,\cdots,p_n)\succ(1/n,1/n,\cdots,1/n).$
whenever $p_i \ge 0$ with $\sum_{i=1}^n p_i = 1.$ Thus applying Proposition \ref{prop:multnom} yields $E_p[Z] \ge E_{(1/n,\cdots,1/n)}[Z] \ge k_0\frac{\log n}{\log\log n}.$
\end{proof}

\begin{figure}
\centering
\hspace{-0.5cm}
\begin{minipage}{0.25\textwidth}
\includegraphics[width=1\textwidth, height = 0.8\textwidth]{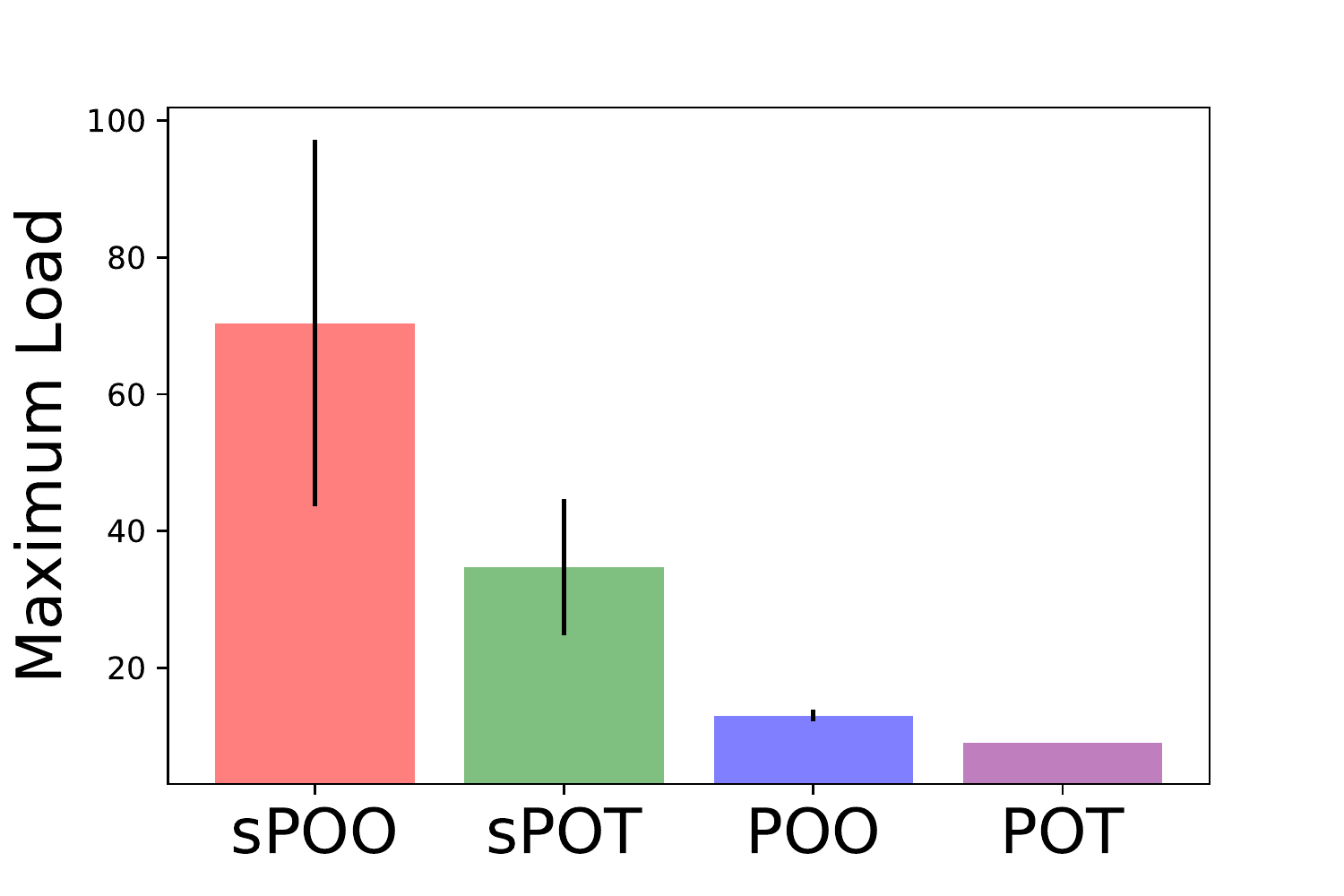}
\subcaption{Maximum load}
\end{minipage}
\hspace{-0.08in}
\begin{minipage}{0.25\textwidth}
\includegraphics[width=0.9\textwidth, height = 0.8\textwidth]{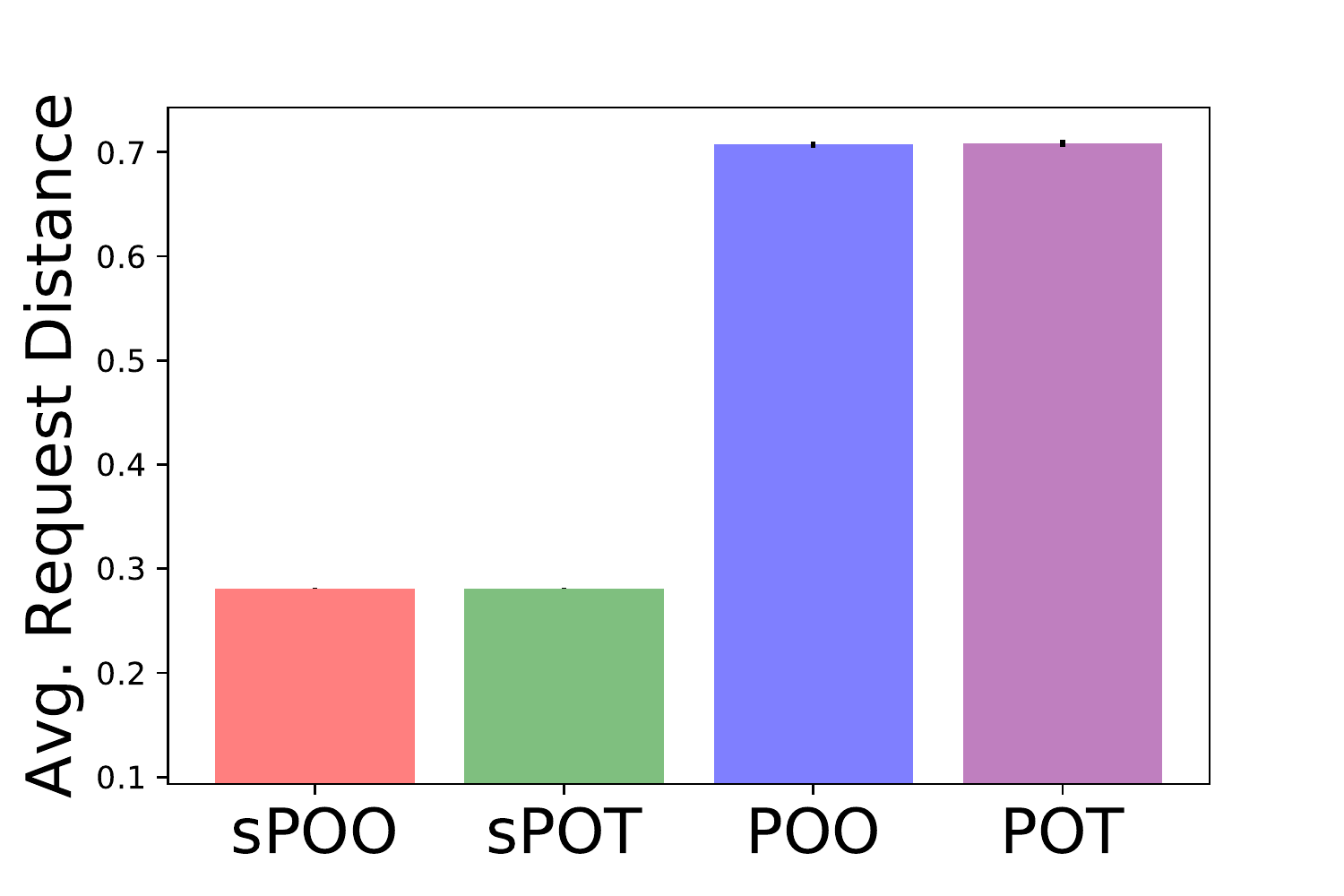}
\subcaption{Expected request distance}
\end{minipage}
\caption{Performance comparison of basic allocation policies  wrt (a) maximum load and (b) expected request distance for $n = 10000$ servers.}
\label{fig:alloc}
\vspace{-0.2in}
\end{figure}
\begin{theorem}\label{th:unispot}
Suppose both users and servers are placed independently and uniformly at random on $\mathcal{D}$. Under sPOT policy, the expected maximum load over all servers is  at least $\Omega(\frac{\log n}{\log\log n})$ with high probability of $1-1/n^{\Omega(1)},$ i.e., we do not get POT benefits.
\end{theorem}
\begin{proof}
Consider the second order Voronoi diagram: $H_S^{(2)}$ associated with the set of servers $S.$ W.l.o.g. consider a cell $\{s_i, s_j\}$ in $H_S^{(2)}$. The probability that a user selects the server pair $\{s_i, s_j\}$ as its two nearest servers is proportional to the area of the cell $\{s_i, s_j\}$. However, the area distribution of cells in $H_S^{(2)}$ is non-uniform (say with probability distribution $p$). We can invoke classical balls and bins argument on $H_S^{(2)}$ as follows inspired by the discussion in \cite{kenthapadi06}. We treat each cell in $H_S^{(2)}$ as a bin. Thus by Lemma \ref{lm:cellsinhv2}, there are $O(3n)$ bins (or cells). Each ball (or user) choses a bin (or a cell) from a distribution $p$  and let $L_{p}$ denote the expected maximum asymptotic load across the bins. Let $L_{U}$ denote  the expected maximum asymptotic load across the bins when $O(n)$ balls are assigned to $O(3n)$ bins with each ball choosing a bin uniformly at random. From classical balls and bins theory, $L_{U} = O(\log n/\log\log 3n) = O(\log n/\log\log n).$ Clearly, by Lemma \ref{lm:majorization}, we have $L_{p}  \ge L_{U} = O(\log n/\log\log n).$ Since a cell consists of a server pair, one of the server pair corresponding to the maximum load would have load at least $(1/2)L_{p}.$ Thus the maximum load across all servers would be at least $(1/2)L_{p} \ge O(\log n/\log\log n).$ 
\end{proof}

\subsection{Tradeoff between Load and Request Distance}
In this Section, we discuss the inherent tradeoff between maximum load and expected request distance metric among different allocation policies. We evaluate the performance of sPOT as compared to other allocation policies. We consider $n = 10000$ servers and an equal number of users placed at a unit square uniformly at random. We ran $10$ trials for each policy. We compare the performance of various allocation policies in Figure \ref{fig:alloc}. 

First, note that with respect to maximum load, the spatial based policies perform worse compared to their classical counterparts. Note that the introduction of spatial aspects into a policy increases its maximum load. For example, sPOT performs worse than both POO and POT as shown in Figure \ref{fig:alloc} (a). Since the maximum asymptotic load for POO is $O(\log n/\log \log n)$ with high probability, Figure \ref{fig:alloc} (a) validates our lower bound results obtained for sPOT in Theorem \ref{th:unispot}.


However, the expected request distance is smallest for sPOO and almost similar to that of sPOT. Also, both POT and POO have very high and similar expected request distances as shown in Figures \ref{fig:alloc} (b). Both results shown in Figure \ref{fig:alloc} (a) and (b) combined, illustrate the tradeoff between maximum load and expected request distance metric. 
%
\section {Improving Load and Request Distance Behavior}\label{sec:kspot}
\begin{figure}[t]
\centering
\hspace{-0.5cm}
\begin{minipage}{0.25\textwidth}
\includegraphics[width=1\textwidth, height = 0.8\textwidth]{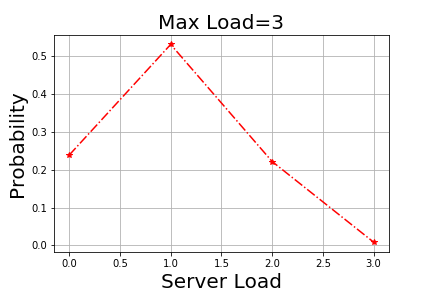}
\subcaption{dPOT Load distribution}
\end{minipage}
\hspace{-0.08in}
\begin{minipage}{0.25\textwidth}
\includegraphics[width=1\textwidth, height = 0.8\textwidth]{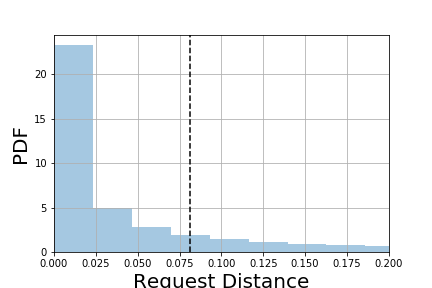}
\subcaption{dPOT distance distribution}
\end{minipage}
\hspace{-1.72in}
\begin{minipage}{0.25\textwidth}
\includegraphics[width=1\textwidth, height = 0.8\textwidth]{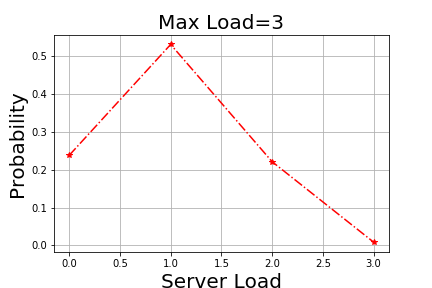}
\subcaption{POT Load distribution}
\end{minipage}
\hspace{-0.21in}
\begin{minipage}{0.25\textwidth}
\includegraphics[width=1\textwidth, height = 0.8\textwidth]{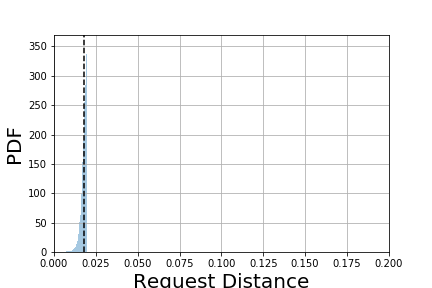}
\subcaption{sPOT distance distribution}
\end{minipage}
\begin{minipage}{0.25\textwidth}
\includegraphics[width=1\textwidth, height = 0.8\textwidth]{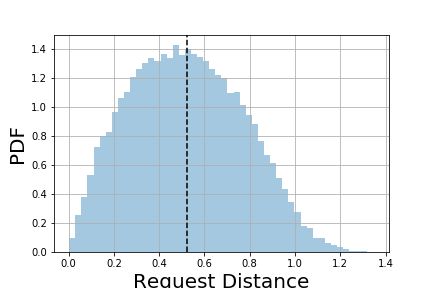}
\subcaption{POT distance distribution}
\end{minipage}
\caption{Performance comparison of allocation policies  wrt dPOT for $n = 50000$ servers. (a) and (b) plots are for dPOT while (c),(e) and (d) for POT and sPOT respectively.}
\label{fig:dPOT}
\vspace{-0.7cm}
\end{figure}
In Section \ref{sec:spot}, we showed that for both grid and uniform based server placement, sPOT does not provide POT benefits. As POT is oblivious to the spatial aspect of user and server distributions, it performs worse with respect to the expected request distance metric. Thus there exists a tradeoff between maximum load and expected request distance among different allocation policies. 

Note that for each user, once its arrival location is fixed, the sampling of two servers in sPOT policy is deterministic while it is completely random for POT. This random sampling over the entire set of servers results in better load behavior for POT  than for sPOT. However, since random sampling in POT is oblivious to the  distances of servers from the particular user, POT incurs very large expected request distance. Thus if one can design a policy with random and distance dependent sampling of servers, such a policy should provide benefits of both POT and sPOT in terms of maximum load and expected request distance respectively. Below we propose and evaluate two such  policies to get benefits of both POT and sPOT. We empirically show that they achieve both POT like load benefits while having a request distance profile similar to that of sPOT.
\subsection{Decay based POT ({\it dPOT})}

Consider the allocation of a random user $r_i$ in the service network. We propose a decay based POT (dPOT) policy to allocate $r_i$ as follows. Under dPOT, $r_i$ samples two servers from $S$ (without replacement), each with probability proportional to $1/d_{ip}^{2}$. Here, $d_{ip}$ denotes the euclidean distance between user $r_i$ and server $s_p.$ $r_i$ then gets allocated to the server with the least load among sampled servers. This rule is similar to the one used in small world routing~\cite{Kleinberg00}. Note that, since the sampling probability of a server is inversely proportional to its distance from $r_i,$ dPOT incurs a smaller expected request distance compared to POT. Surprisingly, dPOT achieves similar load behavior to that of POT. We compare the performance of dPOT to sPOT and POT as follows.

We perform a single simulation run for each of the policies: dPOT, sPOT, POT and measure the distributions of load values across all the servers and of the request distance. Figure \ref{fig:dPOT} (a) shows the load distribution and  Figure \ref{fig:dPOT} (b) shows the request distance distribution for dPOT. We plot the load distribution for POT and request distance distribution for sPOT in Figure \ref{fig:dPOT} (c) and (d) respectively.

First we focus on the server loads in Figure \ref{fig:dPOT} (a) and (c). Interestingly, the load distributions are almost identical for dPOT and POT. Similarly, sPOT performs better than dPOT in terms of request distance distribution as shown in Figure \ref{fig:dPOT} (b) and (d)
since they significantly favor closer nodes. However, compared to POT (as shown in Figure \ref{fig:dPOT} (e)), dPOT performs significantly better in terms of request distances. Thus dPOT achieves the best of both worlds, i.e., low maximum load and low distances. 

\subsection{Candidate set based sPOT ($k$-sPOT)}
We now propose a policy that improves the load behavior of sPOT. We define $C_k$ to be the candidate set (of size $k$) consisting of $k$ nearest servers for a particular user. Under $k$-sPOT, the user selects two servers uniformly at random from $C_k$ and assigns itself to the leastly loaded one. Thus, the formation of candidate set with the $k$ nearest criteria makes the sampling technique distance dependent. Also, the random sampling of two servers within the candidate set helps to balance load and reduce the overall maximum load. Clearly sPOT and POT are two extremes of the policy $k$-sPOT with $k = 2$ and $k = n$ respectively. Below, we discuss the effect of $k$ on maximum load and expected request distance behavior and compare it to other policies.
%

\begin{figure}
\centering
\hspace{-0.2in}
\begin{minipage}{0.22\textwidth}
\includegraphics[width=0.9\textwidth, height = 0.8\textwidth]{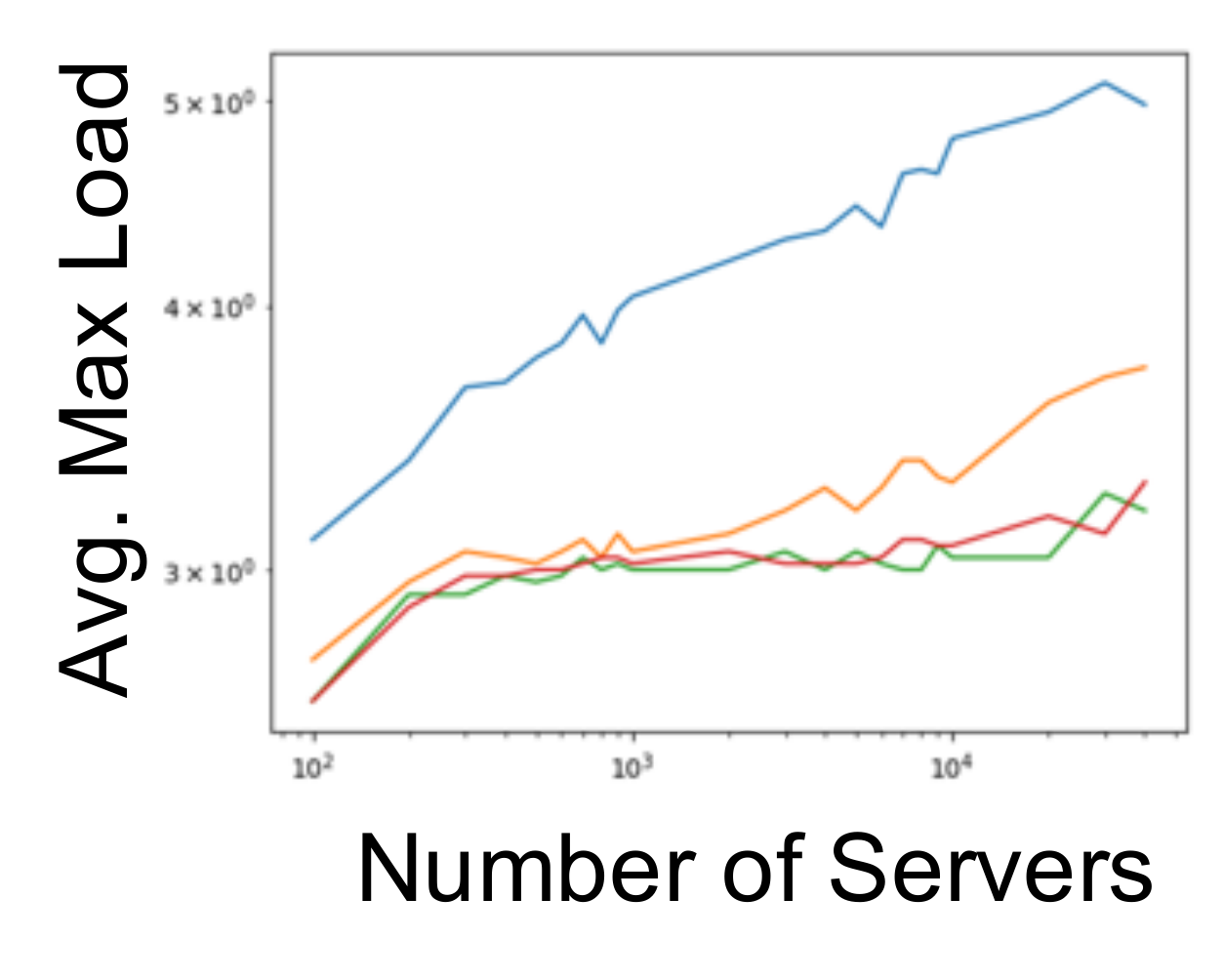}
\subcaption{}
\end{minipage}
\hspace{-0.08in}
\begin{minipage}{0.25\textwidth}
\includegraphics[width=1.2\textwidth, height = 0.8\textwidth]{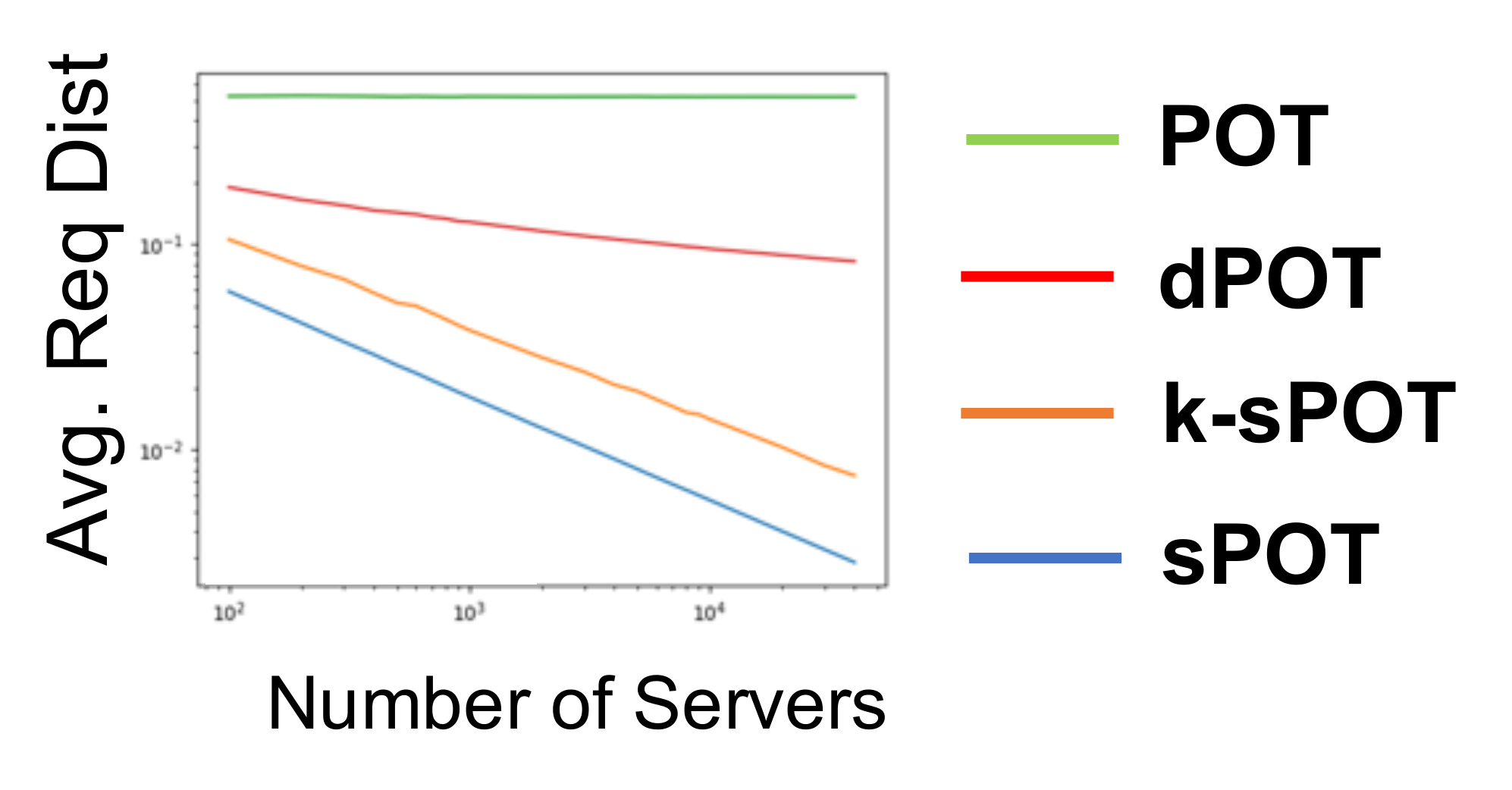}
\subcaption{}
\end{minipage}
\vspace{-0.1in}
\caption{Performance comparison of k-sPOT  with respect to (a) expected maximum load and (b) expected request distance.}
\label{fig:kspot}
\vspace{-0.2in}
\end{figure}

Figure \ref{fig:kspot} (a) shows the growth of average maximum load (averaged over 50 simulation runs for each point) as $n$ is varied from 100 to 40000. We observe that both dPOT and POT perform the best. k-sPOT with $k=\log n$ performs quite well compared to sPOT.  Also, we have observed through simulation that the average maximum load profiles are very similar for k-sPOT with finite k to that of sPOT. Based on these results, we present the following conjecture.
\begin{conjecture}
If the candidate set in k-sPOT does not grow with $n$, no POT benefit is expected.
\end{conjecture}
Figure \ref{fig:kspot} (b) shows how the average request distance drops as $n$ increases (since the node density increases). We observe that, not surprisingly, sPOT outperform the rest. However, k-sPOT with $k=\log n$ performs quite well. Thus k-sPOT with $k=O(\log n)$ achieves good performance for both load and request distance.

\begin{remark}
Note that, since dPOT selects servers through distance based sampling, change in positions of servers theoretically requires choosing a new set of sampling distributions. However, sampling in $k$-sPOT depends on the log neighborhood of the user, thus involves less frequent updates for server sampling distributions.
\end{remark}

\section{Conclusion}\label{sec:con}
%

In this work we considered a class of power of two choices based allocation policy where both resources and users are located on a two-dimensional plane. We analyzed the sPOT policy and provided expressions for the lower bound on the asymptotic maximum load on the resources. We claim that for both grid and uniform based resource placement, sPOT does not provide POT benefits.  We proposed two non-uniform euclidean distance based server sampling policies that achieved the best load and request distance behavior. Experimental results validate the effectiveness of our proposed policies. Finally, going further, we aim at extending our results to consider dynamic arrival of mobile users on the euclidean plane.


\bibliographystyle{abbrv}
\bibliography{refs}  




\end{document}